\documentclass{article}

\usepackage{graphicx}
\usepackage{amsfonts}
\usepackage{amsmath}

\newcommand{\R}{{\mathbb R}}
\newcommand{\Au}{{\bf (A)}}
\newcommand{\Ad}{{\bf (C)}}
\newcommand{\Bu}{{\bf (B)}}
\newcommand{\Bd}{{\bf (C)}}
\newcommand{\somr}{\sum_{r=1}^{k_n}}
\newcommand{\xnunet}{X_{n,1}^{\star}}
\newcommand{\xnret}{X_{n,r}^{\star}}
\newcommand{\ynr}{Y_{n,r}}
\newcommand{\tynr}{\tilde{Y}_{n,r}}
\newcommand{\inr}{I_{n,r}}
\newcommand{\Fnr}{F_{n,r}}
\newcommand{\Dnr}{D_{n,r}}
\newcommand{\Mnr}{M_{n,r}}
\newcommand{\mnr}{m_{n,r}}

\newcommand{\lnr}{\lambda_{n,r}}
\newcommand{\fnchap}{\hat f_n}
\newcommand{\gnchap}{\hat f_n^{G}}
\newcommand{\fnchapx}{\hat f_n(x)}
\newcommand{\fntildx}{\tilde f_n(x)}
\newcommand{\fntild}{\tilde f_n}
\newcommand{\fncheck}{\check f_n}
\newcommand{\fncheckx}{\check f_n(x)}

\def\E{\mathop{\hbox{E}}\nolimits}            
\def\Abs#1{\left\vert #1 \right\vert}   
\def\Var{\mathop{\hbox{Var}}\nolimits}          

\newcommand{\proba}[1]
{P\left(\left\{{#1}\right\}\right)}
\newcommand{\normsup}[1]
{{\left\|{#1}\right\|}_{\infty}} 
\newcommand{\normC}[1]
{{\left\|{#1}\right\|}_{\infty}^C} 
\newcommand{\normun}[1]
{{\parallel{#1}\parallel}_1} 
\newcommand{\normdeux}[1]
{{\parallel{#1}\parallel}_2} 
\newcommand{\grando}[1]
{O\left({#1}\right)} 
\newcommand{\petito}[1]
{o\left({#1}\right)} 
\newcommand{\indic}[1] 
{{\bf 1}_{#1}}

\newtheorem{Theo}{Theorem}
\newtheorem{prop}{Proposition}
\newtheorem{Coro}{Corollary}
\newtheorem{Lem}{Lemma}

\begin{document}

\title{Extreme values and kernel estimates\\ of point processes boundaries}
\author{St\'ephane Girard \& Pierre Jacob\\\\
Laboratoire de Probabilit\'es et Statistique, Universit\'e Montpellier 2,\\
place Eug\`ene Bataillon, 34095 Montpellier cedex 5, France.\\
{\tt \{girard, jacob\}@stat.math.univ-montp2.fr}}

\date{}

\maketitle
\begin{abstract}
We present a method for estimating the edge of a two-dimensional
bounded set, given a finite random set of points drawn from the interior.
The estimator is based both on a Parzen-Rosenblatt kernel and 
extreme values of point processes.  We give conditions
for various kinds of convergence and asymptotic normality.
We propose a method of reducing the negative bias and edge effects, illustrated by a simulation.
\end{abstract}

\section{Introduction}

We address the problem of estimating a bounded set $S$ of $\R^2$ given
a finite random set $N$ of points drawn from the interior. This kind of problem
arises in various frameworks such as
classification~\cite{HarRas}, image processing~\cite{KorTsy3} or econometrics
problems~\cite{DepSimTul}. A lot of different solutions were proposed
since~\cite{Geff1} and~\cite{RenSul} depending on the properties of
the observed random set $N$ and of the unknown set $S$.
In this paper, we focus on the special case 
where $S=\{(x,y)\in \R^2\mid 0\leq x\leq 1~;\;0\leq y \leq f(x)\}$,
with $f$ an unknown function. Thus, the estimation of the subset $S$
reduces to the estimation of the function $f$. 
This problem arises for instance in econometrics where the function $f$ is
called the production frontier.
This is the case in~\cite{Produc}
where data consist of pairs $(X_i,Y_i)$,
$X_i$ representing the input (labor, energy or capital) used to produce
an output $Y_i$ in a given firm $i$. In such a framework, the value $f(x)$
 can be interpreted as the maximum level of output which is attainable 
for the level of input $x$.

Most papers on support estimation use to consider the random set of point $N$
appearing under the frontier $f$ as a $n$-sample. However, in practice, the 
{\em number} as well as the position of the points is random, so we do prefer for 
a long time to deal with point processes. Cox  processes are known to 
provide a high level of generality among the point processes on a plane.
However, after conditioning the intensity, the realization of a Cox 
process is merely the one of a Poisson point process,  so what is really 
observed is a Poisson point process. Moreover in most applications such as 
medical imaging
$f$ delimits a frontier between two zones. A contrasting substance is 
spread on the whole domain, for instance the brain. The magnetic resonance 
imaging only displays the bleeding.
So the healthy part acts as a mask. Inversely, but similarly, when 
investigating the retina, the patient does not detect the small luminous 
spots pointed on a destroyed area. In such cases, there is no way to 
consider the remaining observed points as a random sample. In fact, 
such  truncated empirical point processes are no longer random samples but 
binomial point processes (see \cite{Reiss}).
In fact, even the nature is unable to obtain a random sample on $S$ in this way!
It turns out that binomial point processes are well approximated by Poisson 
processes. Moreover, truncated Poisson point processes are still Poisson 
point processes and the same is true for general Cox processes.
Naturally, as our point of view is not prevailing, we have to preserve the 
possibility of comparing our results with those of authors dealing with 
random samples. So in place of a uniform $n$-empirical process on $S$ with the 
distribution $\lambda/\lambda(S)$, we consider a Poisson point process 
with the intensity $n\lambda/\lambda(S)$,
where $\lambda$ denotes the Lebesgue measure.  
The intensities of the 
two processes are obviously equal. Finally, we claim that we are able to deduce 
for samples similar results  by means of Poisson 
approximations. But it not so simple to achieve, and we prefer to defer 
this work to a further paper.

In the wide range of nonparametric functional estimators~\cite{Bosq}, piecewise
polynomials have been especially studied~\cite{KorTsy,KorTsy3}
and their asymptotic optimality is established under different 
regularity assumptions on $f$.
See~\cite{Har, Jac1,MamTsy} for other cases.
Estimators of $f$ based upon orthogonal series appear in~\cite{AbbSuq,JacSuq}.
In the case of Haar and $C^1$ bases, extreme values estimates are defined 
and studied in~\cite{GirJac, GirJac2, L1Haar} and reveal better properties
than those of~\cite{JacSuq}.
In the same spirit, a Faber-Shauder estimate is proposed in~\cite{Gardes}.
Estimating $f$ can also been considered as a regression problem 
$Y_i=f(X_i)+\varepsilon_i$ with negative noise $\varepsilon_i$.   In this context,
local polynomial estimates are introduced, see~\cite{KK}, or~\cite{Hall} for
a similar approach.

Here a kernel method is proposed in order to obtain smooth estimates 
$\fnchap$.   From the practical point of view, these estimates enjoy
explicit forms and are thus easily implementable.  
From the theoretical point of view, 
we give limit laws with explicit speed
of convergence $\sigma_n$ for 
$\sigma_n^{-1}(\fnchap-\E\fnchap)$ and even for $\sigma_n^{-1}(\fnchap-f)$
after reducing the bias. The rate of convergence of the $L_1$ norm
is proved to be $O(n^{-\frac{\alpha}{5/4+\alpha}})$ for a $\alpha$-Lispchitzian frontier $f$,
which is slightly suboptimal compared to the minimax rate $n^{-\frac{\alpha}{1+\alpha}}$.  
Section~\ref{defs}
is devoted to the definition of the estimator and basic properties of extreme values.
Section~\ref{sectionbias} deals with {\it ad hoc} adaptation of Bochner approximation
results. In Section~\ref{sectionestim}, we give the main results of convergence:
mean square uniform convergence and almost complete uniform convergence. We prove,
in Section~\ref{sectionasymp}, the asymptotic normality of the estimator, when
centered to its mathematical expectation. Section~\ref{sectionreduc} is devoted to
some bias reductions, allowing in certain cases asymptotic normality for an 
estimator, when centered to the function $f$. We also present a technique for avoiding
edge effects. In~\cite{Nous}, a simulation gives an idea of the improvements
carried off by these modifications.
Section~\ref{seccomp} is dedicated to 
comparison of kernel estimates with the other propositions found
in the literature.

\section{Definition and basic properties}
\label{defs}

For all $n>0$, let $N$ be a Poisson point process with mean measure 
$nc\lambda$,
where $\lambda$ denotes
the Lebesgue measure on a subset $S$ of $\R^2$ defined as follows:
\begin{equation}
\label{defS}
 S=\{(x,y)\in {\R}^2\mid 0\leq x\leq 1~;\;0\leq y \leq f(x)\}.
\end{equation}
The normalization parameter $c$ is defined by $c=1/\lambda(S)$ such that
$\E(N(S))=n$.  
We assume that on $[0,1]$, $f$ is a bounded measurable function, strictly 
positive and $\alpha$-Lipschitz, $0<\alpha\leq 1$, with Lipschitz 
multiplicative constant $L_f$ and that $f$ vanishes elsewhere. 
We denote by $m$ (and $M$) the lower (and the upper)
bound of $f$ on $[0,1]$.
Given $(h_n)$ a sequence of positive real numbers such that $h_n\to 0$
when $n\to\infty$, the function $f$ is approximated by the convolution:
\begin{equation}
\label{gnx}
g_n(x)  =\int_{\R} K_n(x-y) f(y)dy,\;\; x\in[0,1],
\end{equation}
where $K_n$ is given by
$$
K_n(t)=\frac{1}{h_n}K\left(\frac{t}{h_n}\right), \;\; t\in\R,
$$
and $K$ is a bounded positive Parzen-Rosenblatt kernel {\it i.e.} verifying:
$$
\forall x\in\R, \; 0\leq K(x)\leq \sup_{\R}K<+\infty,\quad
\int_{\R} K(t)dt=1, \quad 
\lim_{\Abs{x}\to\infty} x K(x) = 0.
$$
Note that $K^2$ and $K^3$ are Lebesgue-integrable.
In the sequel, we introduce extra hypothesis on $K$ when necessary.
Consider $(k_n)$ a sequence of integers increasing to infinity 
and divide $S$ into $k_n$ cells $D_{n,r}$ with:
$$
\Dnr = \{\;(x,y) \in S \mid x\in \inr\;\}, \qquad
\inr = \left[\frac{r-1}{k_n}, \frac{r}{k_n}\right[,\quad r=1,\ldots,k_n.
$$
The convolution~(\ref{gnx}) is discretized on the $\{\inr\}$ subdivision
of $[0,1]$:
$$
f_n(x)  =  \frac{1}{k_n}\somr K_n(x-x_r) f(x_r), \;\; x\in[0,1],
$$
where $x_r$ is the center of $\inr$.
The values $f(x_r)$ of the function on the subdivision are estimated
through $\xnret$ the supremum of the second coordinate of
 the points of the truncated process $N(.\cap \Dnr)$.
The considered estimator can be written as:
\begin{equation}
\label{estimate}
\fnchapx  =\frac{1}{k_n}\somr K_n(x-x_r) \xnret, \;\; x\in[0,1].
\end{equation}
Formally, this estimator is very similar to the estimators based on
expansion of $f$ on $L^2$ bases~\cite{GirJac,GirJac2} although it
is obtained by a different principle.
Besides, combining the uniform kernel $K(t)=\indic{[-1/2,1/2]}(t)$
with the bandwidth $h_n=1/k_n$ yields Geffroy's estimate:
\begin{equation}
\label{estiGeff}
\gnchap(x)  =\somr \indic{\inr}(x) \xnret, \;\; x\in[0,1],
\end{equation}
which is piecewise constant on the $\{\inr\}$ subdivision of $[0,1]$.  
At the opposite, here we focus on smooth estimators
obtained by considering smooth kernels in~(\ref{estimate}).
More precisely, we examine systematically the convergence properties of the estimator in
two main situations:
\begin{description}
\item [\Au] $K$ is $\beta$-Lipschitz on $\R$, $0<\beta\leq 1$, 
with Lipschitz multiplicative constant $L_K$, 
$x\rightarrow x^2K(x)$ is integrable, $k_n=\petito{n}$,
$h_n k_n^{\alpha}\to\infty$, and $h_n^{1+\beta} k_n^{\beta}\to\infty$ when $n\to \infty$.
\item [\Bu] $K$ has a compact support, a bounded first derivative and
is piecewise $C^2$,
$k_n=\petito{n}$ and $h_n k_n\to\infty$ when $n\to \infty$.
\end{description}
Of course, Geffroy's estimate does not fulfil these conditions.  
Some stochastic convergences will require extra conditions on the $(k_n)$ sequence:
\begin{description}
\item [\Ad] $k_n=\petito{{n}/{\ln n}}$ and $n=\petito{k_n^{1+\alpha}}$.
\end{description}
Throughout this paper, we write:
$$
\lambda (D_{n,r})=\lnr,\; \min_{x\in\inr} f(x) = \mnr,\; \max_{x\in\inr} f(x) = \Mnr.
$$
The cumulative distribution function of $\xnret$ is easily calculated
on $[0,\mnr]$, after noticing that, for every measurable $B\subset S$,
$P(N(B)=0)=\exp{(-nc\lambda(B))}$:
\begin{equation}
\label{eqfdr}
\Fnr(x)=P(\xnret\leq x)=\exp{\left(\frac{nc}{k_n}(x-k_n\lnr)\right)}, \;\; x\in [0,\mnr].
\end{equation}
Of course, $\Fnr(x)=0$ if $x<0$ and $\Fnr(x)=1$ if $x>\Mnr$. For $x\in [\mnr,\Mnr]$,
$\Fnr(x)$ is unknown, but $1-\Fnr(\mnr)$ can be controlled through regularity conditions
made on $f$. Finally,~(\ref{eqfdr}) and this control provide precise expansions for the
first moments of $\xnret$. We quote that useful results in the following lemma.

\begin{Lem}
 \label{lemespe} 
Assume \Ad~is verified. Then,
\begin{description}
\item [(i)] $\displaystyle \max_r \Abs{\E(\xnret)- k_n\lnr +\frac{k_n}{nc}} = \grando{\frac{n}{k_n^{1+2\alpha}}}$,
\item [(ii)] $\displaystyle \max_r \Abs{\Var(\xnret)- \frac{k_n^2}{n^2 c^2}} = \grando{\frac{1}{k_n^{2\alpha}}}$,
\item [(iii)] $\displaystyle \max_r \E\left(\Abs{\xnret-E(\xnret)}^3\right)= \grando{\frac{k_n^3}{n^3}}$.
\end{description}
\end{Lem}
We shall also need a lemma on the Parzen-Rosenblatt kernel. 
\begin{Lem}
\label{lemK}
Let $a\neq 0$. For any probability sequence $(P_n)$, we have 
\begin{equation}
\label{eqlemK}
\int K(u)K\left(u+\frac{a}{h_n}\right) P_n(du) =  \petito{h_n}.
\end{equation}
\end{Lem}
Proof of all lemmas are postponed to the Appendix.  
\begin{Coro}
\label{coroK}
For all $x\neq y$,
$$
\frac{1}{h_nk_n}\somr K\left(\frac{x-x_r}{h_n}\right) K\left(\frac{y-x_r}{h_n}\right) = \petito{1}.
$$
\end{Coro}
This result is deduced from Lemma~\ref{lemK} with $a=y-x$ and
$\displaystyle P_n=\frac{1}{k_n}\somr \delta_{\frac{x-x_r}{h_n}}$.

\section{Bias convergence}
\label{sectionbias}

We first give conditions on the sequences $(h_n)$ and $(k_n)$ to
obtain the local uniform convergence of $f_n$ to $f$, that is,
the uniform convergence on every compact subset $C$ of $]0,1[$.
Of course, since $f$ is not continuous at 0 and 1, we cannot
obtain uniform convergence on the whole compact $[0,1]$.  
We note in the sequel:
$$
\normC{g}=\sup_{x\in C} \Abs{g(x)},
$$
for all function $g:[0,1]\to\R$.
The triangular inequality
$$
\normC{f-f_n} \leq \normsup{f_n-g_n} + \normC{g_n-f},
$$
shows the two contributions to the bias. 
The first term, studied in Lemma~\ref{lembiais1}, is a consequence
of the discretization of~(\ref{gnx}).
The second term is studied in Lemma~\ref{lembiais2}. It appears 
in various other kernel estimates such as regression 
or density estimates.

\begin{Lem}
\hfill
\label{lembiais1}
\begin{description}
\item [(i)] Under \Au,
$\displaystyle \normsup{f_n-g_n}=\grando{\frac{1}{h_n k_n^\alpha}}+\grando{\frac{1}{h_n^{1+\beta} k_n^\beta}}$.
\item [(ii)] Under \Bu,
$\displaystyle \normsup{f_n-g_n}=\grando{\frac{1}{h_n^2 k_n^2}}+
\grando{\frac{1}{k_n^\alpha}}$.
\end{description}
\end{Lem}
The function $f$ is uniformly continuous on $[0,1]$ as soon as it is continuous on the same compact interval
and the Bochner Lemma entails that $\normC{g_n-f}\to 0$ as $n\to\infty$.
The following lemma precises this result by providing the rates of the convergence of $\normC{g_n-f}$ in different situations. 
\begin{Lem} 
\label{lembiais2}
\hfill
\begin{description}
\item [(i)] If $x\to x^2 K(x)$ is integrable, then
$\displaystyle \normC{g_n-f}=\grando{h_n^{\frac{2\alpha}{\alpha+2}}}$.
\item [(ii)] If $K$ has a compact support then
$\displaystyle \normC{g_n-f}=\grando{h_n^\alpha}.  $
\end{description}
\end{Lem}
\noindent Let us note that in situation \Bu, $1/k_n^\alpha=o(h_n^\alpha)$.  
Thus, as a simple consequence of Lemma~\ref{lembiais1} and
Lemma~\ref{lembiais2}, we get:
\begin{prop}
\label{propbiais}
\hfill
\begin{description}
\item [(i)] Under \Au :
$\displaystyle \normC{f_n-f}=\grando{\frac{1}{h_n k_n^\alpha}}+\grando{\frac{1}{h_n^{1+\beta} k_n^\beta}}
+ \grando{h_n^{\frac{2\alpha}{\alpha+2}}}$.
\item [(ii)] Under \Bu :
$\displaystyle \normC{f_n-f}=\grando{\frac{1}{h_n^2 k_n^2}}+\grando{h_n^\alpha}$.
\end{description}
In either case, $f_n$ converges uniformly locally to $f$. 
\end{prop}
Applying Proposition~\ref{propbiais} to the function $\indic{[0,1]}$ leads to
the following corollary which will reveal useful in the following.
\begin{Coro}
\label{corosom}
Under the conditions of Proposition~\ref{propbiais},
$$
\lim_{n\to\infty} \normC{\frac{1}{k_n}\somr K_n(.-x_r) -1 } = 0.
$$
\end{Coro}

\section{Estimate convergences }
\label{sectionestim}

This section is devoted to the study of the stochastic convergence of $\fnchap$
to $f$. We establish sufficient conditions for 
mean square local uniform convergence and almost complete local uniform convergence.

\subsection{Mean square local uniform convergence }

In this paragraph, we give sufficient conditions for
$$
\sup_{x\in C} \E\left[(\fnchapx-f(x))^2\right] \to 0 \mbox{ as } n \to \infty,
$$
where $C$ is compact subset of $]0,1[$.
The well-known expansion
$$
\E\left[(\fnchapx-f(x))^2\right]=\left[\E(\fnchapx)-f(x)\right]^2+\Var(\fnchapx)
$$
allows one to consider the bias term and the variance term separately.
The two following lemmas are devoted to the bias which splits in turn as
$$
\normC{\E(\fnchap)-f}\leq \normC{\E(\fnchap)-f_n} + \normC{f_n-f}.
$$
\begin{Lem} 
\label{lembiais3}
Suppose $\Bd$ is verified.  Under \Au~or \Bu:
$\displaystyle \normC{\E(\fnchap)-f_n}=\grando{{k_n}/{n}}$.
\end{Lem}
\noindent As a consequence of Lemma~\ref{lembiais3} and Proposition~\ref{propbiais}, we obtain
the behavior of the bias:
\begin{Lem} 
\label{lembiais4}
Suppose $\Bd$ is verified.  
\begin{description}
\item [(i)] Under \Au:
$$
\normC{\E(\fnchap)-f}=\grando{\frac{k_n}{n}}+\grando{\frac{1}{h_n k_n^\alpha}}+\grando{\frac{1}{h_n^{1+\beta} k_n^\beta}} + \grando{h_n^{\frac{2\alpha}{\alpha+2}}}.
$$
\item [(ii)] Under \Bu:
\begin{equation}
\label{biaisori}
\normC{\E(\fnchap)-f}=\grando{\frac{k_n}{n}}+\grando{\frac{1}{h_n^2 k_n^2}}+\grando{h_n^\alpha}.
\end{equation}
\end{description}
\end{Lem}
To conclude, it remains to consider the variance term.
\begin{Lem} 
\label{lemvariance}
Suppose $\Bd$ is verified.  Under \Au~or \Bu:
$$
\lim_{n\to\infty} \normC{\frac{\Var(\fnchap)}{\sigma_n^2} - \sigma^2} = 0,
$$
where ${\sigma}_n=\frac{k_n^{1/2}}{n h_n^{1/2}}$ 
and $\sigma=\frac{\normdeux{K}}{c}$.
\end{Lem}
\noindent In situation \Bu, $\sigma_n=o(k_n/n)$, and therefore
the variance of the estimator is small with respect to the bias.  
In both situations, as a consequence of Lemma~\ref{lembiais4} and Lemma~\ref{lemvariance}, we get:
\begin{Theo}
\label{cvumq}
Suppose $\Bd$ is verified.  
\begin{description}
\item [(i)] Under \Au:
$$
\normC{\E(\fnchap-f)^2}=
\grando{\frac{k_n^2}{n^2}}+\grando{\frac{1}{h_n^2 k_n^{2\alpha}}}+\grando{\frac{1}{h_n^{2+2\beta} k_n^{2\beta}}} + \grando{h_n^{\frac{4\alpha}{\alpha+2}}} + \grando{\frac{k_n}{n^2 h_n}}.
$$
\item [(ii)] Under \Bu:
$$
\normC{\E(\fnchap-f)^2}= \grando{\frac{k_n^2}{n^2}}+\grando{\frac{1}{h_n^4 k_n^4}}+\grando{h_n^{2\alpha}}.
$$
\end{description}
In either case, the mean square local uniform convergence of $\fnchap$ to $f$ follows.
\end{Theo}
In situation \Bu, choosing $k_n=n^{\frac{\alpha+2}{3\alpha+2}}$
and $h_n=n^{-\frac{2}{3\alpha+2}}$ yields 
$$
\normC{\E(\fnchap-f)^2}=\grando{n^{-\frac{4\alpha}{3\alpha+2}}},
$$
and thus, we obtain the following bound for the $L_1$ norm:
\begin{equation}
\label{vitessebiaise}
\E\left(\normun{\fnchap-f}\right)=\E\left(\int_0^1 \Abs{\fnchap(x)-f(x)}dx\right)
\leq \left(\normC{\E(\fnchap-f)^2}\right)^{1/2}=\grando{n^{-\frac{\alpha}{1+\frac{3}{2}\alpha}}}.
\end{equation}
As a comparison, the minimax rate in the $n$-sample case is
$n^{-\frac{\alpha}{1+ \alpha }}$ and is reached by Geffroy's estimate.  
A bias reduction method
will be introduced in Section~\ref{sectionreduc} in order to ameliorate the
bound~(\ref{vitessebiaise}).

\subsection{Almost complete local uniform convergence  }

We shall give sufficient conditions for the convergence of the series
$$
\forall \varepsilon >0, \qquad
\sum_{n=1}^{+\infty}\proba{\normC{\fnchap-f}>\varepsilon
}<+\infty.
$$
\begin{Theo}
\label{cvupco}
Suppose $k_n=o(n/\ln n)$.
Under \Au~or \Bu,
$\fnchap$ is almost completely locally uniformly convergent to $f$.
\end{Theo}
\begin{proof}
Let $C$ be a compact subset of $]0,1[$ and $\varepsilon>0$.
From Proposition~\ref{propbiais}, $f_n$ converges uniformly to $f$ on $C$.
It remains to consider $\normC{\fnchap-f_n}$. For $x\in C$, we have:
\begin{eqnarray*}
\Abs{\fnchapx-f_n(x)} &\leq& \frac{1}{k_n} \somr K_n(x-x_r) \max_r\Abs{\xnret - f(x_r)} \\
&\leq& \left(1+\normC{\frac{1}{k_n} \somr K_n(.-x_r)-1} \right) \max_r\Abs{\xnret - f(x_r)}\\
&\leq& \left(1 +\petito{1}\right) \max_r\Abs{\xnret - f(x_r)},
\end{eqnarray*}
with Corollary~\ref{corosom}.
Now, since $f$ is continuous on $[0,1]$, $\Mnr - \mnr < \varepsilon/2$ uniformly in $r$,
for $n$ large enough, and therefore 
$$
\left\{\max_r\Abs{\xnret-f(x_r)}> \varepsilon \right\} \subset
\bigcup_r \left\{f(x_r)-\xnret>\varepsilon \right\} \subset
\bigcup_r \left\{\xnret<\mnr-\varepsilon/2 \right\}.
$$
As a consequence,
$$
\proba{\max_r\Abs{\xnret-f(x_r)}> \varepsilon} \leq \somr \Fnr(\mnr-\varepsilon/2 ),
$$
where $\Fnr$ is given by~(\ref{eqfdr}). Then, the inequality
$$
\proba{\max_r\Abs{\xnret-f(x_r)}> \varepsilon} \leq
 k_n \exp{\left(-\frac{nc\varepsilon}{2k_n}\right)}
$$
entails the convergence of the series with $k_n=o(n/\ln n)$.  
\end{proof}

\section{Asymptotic distributions}
\label{sectionasymp}

In Theorem~\ref{thnorasymp}, we give the limiting distribution of the random variable $\fnchapx$
for a fixed $x\in C$, a compact subset of $]0,1[$. In Theorem~\ref{thnorasymp2},
 we study the asymptotic distribution
of the random vector obtained by evaluating $\fnchap$ in several distinct points
of $C$.

\begin{Theo}
\label{thnorasymp}
Suppose \Ad~is verified.
Under \Au~or \Bu,
$s_n(x)={\sigma}_n^{-1}(\fnchapx - \E(\fnchapx))$
converges in distribution to a centered Gaussian variable with variance $\sigma^2$,
for all $x\in C$.
\end{Theo}
\begin{proof}
Let $x\in C$ be fixed. Introducing the $k_n$ independent random variables
$$
\ynr(x)=\frac{n}{k_n^{3/2}h_n^{1/2}}K\left(\frac{x-x_r}{h_n}\right) (\xnret-\E\xnret),
$$
the quantity $s_n(x)$ can be rewritten as
$$
s_n(x)=\somr \ynr(x).
$$
Our goal is to prove that the Lyapounov condition
$$
\lim_{n\to\infty} \somr \frac{\E\left(\Abs{\ynr(x)}^3\right)}{\Var^{3/2}(s_n(x))} = 0
$$
holds under condition \Au, \Ad~or \Bu, \Bd. Remark first that
$$
\Var(s_n(x))={ \Var(\fnchapx) }/{\sigma_n^2} \to \sigma^2,
$$
as $n\to\infty$ with Lemma~\ref{lemvariance}. Second, we have
\begin{eqnarray*}
\somr \E\left(\Abs{\ynr(x)}^3\right)&\leq &\frac{n^3}{h_n^{3/2}k_n^{9/2}} \somr K^3\left(\frac{x-x_r}{h_n}\right) \max_r \E\left(\Abs{\xnret-E(\xnret)}^3\right)\\
&\leq & \left(\frac{1}{k_nh_n} \somr K^3\left(\frac{x-x_r}{h_n}\right)\right) \grando{\frac{1}{k_n^{1/2} h_n^{1/2}}},
\end{eqnarray*}
with Lemma~\ref{lemespe}. Then,
\begin{eqnarray*}
\frac{1}{h_n k_n} \somr K^3\left(\frac{x-x_r}{h_n}\right)&\leq& \int_{\R} K^3(u)du + \normC{\frac{1}{h_n k_n} \somr K^3\left(\frac{.-x_r}{h_n}\right) - \int_{\R} K^3(u)du}\\
&\leq& \int_{\R} K^3(u)du +\petito{1},
\end{eqnarray*}
with Corollary~\ref{corosom} applied to the kernel $K^3/\int K^3(u)du$. As a conclusion,
$$
\somr \frac{\E\left(\Abs{\ynr(x)}^3\right)}{\Var^{3/2}(s_n(x))} = 
\grando{\frac{1}{h_n^{1/2} k_n^{1/2}}},
$$
and the result follows.
\end{proof}

\begin{Theo}
\label{thnorasymp2}
Let $(y_1,\dots,y_q)$ be distinct points in $C$ and denote ${\mathbb I}_q$
the identity matrix of size~$q$. Under the conditions of Theorem~\ref{thnorasymp},
the random vector $(s_n(y_j),\;j=1,\dots,q)$ 
converges in distribution to a centered Gaussian vector of $\R^q$ with covariance
matrix  $\sigma^2 {\mathbb I}_q$.
\end{Theo}
\begin{proof}
Our goal is to prove that, $\forall (u_1,\dots,u_q)\in\R^q$, the random variable
$$
\tilde{s}_n=\sum_{i=1}^q u_i s_n(y_i)
$$
converges in distribution to a centered Gaussian variable with variance $\normdeux{u}^2\sigma^2$.
A straightforward calculation yields
$$
\tilde{s}_n=\somr \tynr,
$$
where we have defined
$$
\tynr=\frac{n}{k_n^{3/2}h_n^{1/2}}\sum_{i=1}^q u_i K\left(\frac{y_i-x_r}{h_n}\right) (\xnret-\E\xnret).
$$
We use a chain of arguments similar to the ones in Theorem~\ref{thnorasymp} proof.
First, the variance of $\tynr$ is evaluated with Lemma~\ref{lemespe}{\bf (ii)}:
\begin{eqnarray*}
\Var(\tynr) & = & \frac{n^2}{h_n k_n^3} \left(\sum_{i=1}^q u_i K\left(\frac{y_i-x_r}{h_n}\right) \right)^2 \Var(\xnret) \\
&=& \frac{1}{c^2}\frac{1}{h_nk_n} \left(\sum_{i=1}^q u_i K\left(\frac{y_i-x_r}{h_n}\right) \right)^2
\left(1+\grando{\frac{n^2}{k_n^{2\alpha+2}}}\right) \\
&\sim& \frac{1}{c^2}\frac{1}{h_nk_n} \left(\sum_{i=1}^q u_i K\left(\frac{y_i-x_r}{h_n}\right) \right)^2.
\end{eqnarray*}
Then, the variance of $\tilde{s}_n$ can be expanded as
$$
\Var(\tilde{s}_n)\sim \sum_{i=1}^q \somr\frac{u_i^2}{c^2}\frac{1}{h_nk_n}K^2\left(\frac{y_i-x_r}{h_n}\right)+\sum_{i\neq j}\frac{u_i u_j}{c^2}\frac{1}{h_nk_n} \somr K\left(\frac{y_i-x_r}{h_n}\right) 
K\left(\frac{y_j-x_r}{h_n}\right).
$$
Corollary~\ref{corosom} provides the limit of the first term and, from Corollary~\ref{coroK}, the second term goes to 0 when $n$ goes to infinity. As a partial conclusion,
$\Var(\tilde{s}_n)\to\normdeux{u}^2\sigma^2$ when $n\to\infty$.
 Now, we have
$$
\E\left(\Abs{\tynr}^3\right)=\frac{n^3}{h_n^{3/2}k_n^{9/2}} \Abs{\sum_{i=1}^q u_i K\left(\frac{y_i-x_r}{h_n}\right)}^3 \E\left(\Abs{\xnret-E(\xnret)}^3\right).
$$
Then, Lemma~\ref{lemespe}{\bf (iii)} entails
$$
\somr \E\left(\Abs{\tynr}^3\right)  =  \grando{\frac{1}{h_n^{3/2}{k_n^{3/2}}}}
\somr \Abs{\sum_{i=1}^q u_i K\left(\frac{y_i-x_r}{h_n}\right)}^3,
$$
and remarking that
$$
\Abs{\sum_{i=1}^q u_i K\left(\frac{y_i-x_r}{h_n}\right)}^3 \leq \normsup{K}\normun{u}
\left(\sum_{i=1}^q u_i K\left(\frac{y_i-x_r}{h_n}\right)\right)^2
$$
shows finally that
$$
\somr \E\left(\Abs{\tynr}^3\right)  =  \grando{\frac{1}{h_n^{1/2}{k_n^{1/2}}}},
$$
and the conclusion follows. 
\end{proof}

\section{Bias reduction}
\label{sectionreduc}

It is worth noticing that, in Section~\ref{sectionasymp}, the negative bias of $\fnchap$
is too large to obtain a limit distribution for $(\fnchap-f)$. 
We introduce a corrected estimator $\fntild$ sharp enough to obtain a limiting 
distribution for $(\fnchap-f)$ under conditions \Bu$\;$ and \Bd.

\noindent It is clear, in view of Lemma~\ref{lemespe}, that $\xnret$ is an estimator of $k_n\lnr$
with a negative bias asymptotically equivalent to $-{k_n}/{(nc)}$.
To reduce this bias, we introduce the random variable defined by
$$
Z_n=\frac{1}{n-k_n}\somr \xnret.
$$
Lemma~\ref{lemespe} implies that, under~\Bd,
\begin{equation}
\label{espezn}
\E(Z_n)=\frac{k_n}{nc}+\grando{\frac{1}{k_n^{2\alpha}}}.
\end{equation}
This suggests to consider the estimator
$$
\fntildx=\frac{1}{k_n}\sum_{r=1}^{k_n} K_n(x-x_r)\left(\xnret+Z_n\right),\;\; x\in[0,1].
$$
A more precise version of Theorem~\ref{thnorasymp} can be given in situation \Bu~at
the expense of additional conditions.  
To this end, we need a preliminary lemma providing the bias of the new estimator $\tilde{f}_n$.
\begin{Lem}
\label{newbias}
Under \Bu, \Bd:
$$
\normC{\E(\tilde{f}_n)-f}=\grando{\frac{1}{h_n^2 k_n^2}}+\grando{h_n^\alpha}.
$$
\end{Lem}
\noindent The bias of the new estimator $\tilde{f}_n(x)$ is asymptotically lower than the bias of $\fnchapx$
since the $k_n/n$ term of~(\ref{biaisori}) is cancelled in Lemma~\ref{newbias}.
Let us also note that the variance of $\tilde{f}_n(x)$ is bounded above
by the variance of $\fnchapx$: Since
\begin{equation}
\label{eqvar1}
\Var\tilde{f}_n(x)\leq 2\Var \fnchapx +2\left(\frac{1}{k_n}\somr K_n(x-x_r)\right)^2\Var Z_n,
\end{equation}
it follows from Lemma~\ref{lemespe}, Lemma~\ref{lemvariance} and 
Corollary~\ref{corosom} that
\begin{equation}
\label{eqvar2}
\frac{\Var\tilde{f}_n(x)}{\Var\fnchap(x)}\leq 2+ \grando{\frac{\Var Z_n}{\sigma_n^2}}
=2 + \grando{\frac{k_n\Var \xnunet}{n^2\sigma_n^2}}
=2 + \grando{\frac{h_nk_n^2}{n^2}} = 2 +\petito{1}.
\end{equation}
These remarks allow  to give the asymptotic distribution of $(\fntildx - f(x))$.
\begin{Theo}
\label{thnorasymp3}
If \Bu~holds, $n=\petito{k_n^{1/2}h_n^{-1/2-\alpha}}$,
$n=\petito{k_n^{5/2}h_n^{3/2}}$ and $k_n=\petito{{n}/{\ln n}}$,
then
$t_n(x)={\sigma}_n^{-1}(\fntildx - f(x))$
converges in distribution to a centered Gaussian variable with variance $\sigma^2$,
for all $x$ in a compact subset of $]0,1[$.
\end{Theo}

\begin{proof}
Consider the expansion
\begin{eqnarray*}
t_n(x)& = &{\sigma}_n^{-1}(\fntildx-\E(\fntildx))+{\sigma}_n^{-1}(\E(\fntildx) - f(x))\\
&=& {\sigma}_n^{-1}(\fnchapx-\E(\fnchapx))+ \frac{{\sigma}_n^{-1}}{k_n} \somr K_n(x-x_r) (Z_n-\E(Z_n)) + {\sigma}_n^{-1} (\E(\fntildx) - f(x)).
\end{eqnarray*}
In view of Theorem~\ref{thnorasymp}, the first term converges in distribution to  
a centered Gaussian variable with variance $\sigma^2$.
The second term is centered and its variance converges to zero from~(\ref{eqvar1})
and~(\ref{eqvar2}).
Therefore, the second term converges to 0 in probability.
The third term is controlled with Lemma~\ref{newbias}:
$$
{\sigma}_n^{-1}\normC{\E(\fntild)-f}=\grando{\frac{n}{h_n^{3/2}k_n^{5/2}}} 
+ \grando{\frac{nh_n^{1/2+\alpha}}{k_n^{1/2}}} \to 0,
$$
and the conclusion follows.
\end{proof}
\noindent The uniform mean square distance between $\fntild$ and $f$ is
derived from Lemma~\ref{lemvariance} and Lemma~\ref{newbias}:
\begin{equation}
\label{eqerreurfin}
\normC{\E(\tilde{f}_n-f)^2}=\grando{\frac{k_n}{n^2 h_n}}+\grando{\frac{1}{h_n^4 k_n^4}}+\grando{h_n^{2\alpha}}.
\end{equation}
\noindent Possible choices are
$k_n=n^{\frac{4+2\alpha}{4+5\alpha}}$ and
$h_n=n^{-\frac{4}{4+5\alpha}}$ 
leading to
$$
\normC{\E(\tilde{f}_n-f)^2}=\grando{n^{-\frac{8\alpha}{4+5\alpha}}},
$$
and thus,
\begin{equation}
\label{vitessenoyau}
\E\left(\normun{\tilde{f}_n-f}\right)= \grando{n^{-\frac{\alpha}{1+\frac{5}{4}\alpha}}},
\end{equation}
which is a significant improvement of (\ref{vitessebiaise}).
It is well-known that non-parametric estimators based on Parzen-Rosenblatt kernels
suffer from a lack of performance on the boundaries of the estimation interval.
To overcome this limitation, symmetrization techniques have been developed~\cite{Cow}.
The application of such a method to $\fntildx$ yields the following estimator:
$$
\fncheckx=\frac{1}{k_n}\sum_{r=1}^{k_n} \left(K_n(x-x_r)+K_n(x+x_r)+K_n(x+x_r-2)\right)\left(\xnret+Z_n\right),\;\; x\in[0,1].
$$
The convergence properties of $\fnchap$ and $\fntild$ on the compact subsets of
$]0,1[$ can be extended to $\fncheck$ on the whole interval $[0,1]$ without
difficulties.

\section{Comparison with other estimates}
\label{seccomp}

Let us emphasize that such comparisons
are only relevant within a same framework,
which excludes hypotheses such as the convexity or the monotonicity of $f$.
Thus, the competitive methods to our kernel approach are essentially
local polynomial estimates~\cite{KK,Hall},
piecewise polynomial estimates~\cite{KorTsy,KorTsy3}
and our projection estimate~\cite{GirJac,GirJac2}. 

-- From the theoretical point of view, piecewise polynomial estimates
benefit from the minimax optimality
whereas the estimates proposed in this paper are suboptimal.
In the class of continuous functions $f$ having a
Lipschitzian $k$-th derivative, the optimal rate of convergence is attained
by minimizing, on each cell of a partition of $[0,1]$, the measure of a 
domain with a polynomial edge of degree $k$. 
For instance, in the case of a $\alpha$-Lipschitzian frontier, the minimax
optimal rate for the $L_1$ norm is $n^{-\frac{\alpha}{1+\alpha}}$ and
the corresponding rate is $n^{-\frac{\alpha}{5/4+\alpha}}$ 
for $\tilde{f}_n$  (see (\ref{vitessenoyau})). The difference of
speed increases with $\alpha$, but even if $\alpha=1$ (which is the
worst situation for us), one obtains ``similar'' rates of convergence,
that is $n^{-1/2}$ and $n^{-4/9}$.  In this sense, kernel estimates
bring a significant improvement to projection estimates.  

-- From the practical point of view, all the previous estimates require
the selection of two hyper-parameters.  
In case of piecewise polynomial and local polynomial estimators,
the construction of the estimate requires to select the degree of
the polynomial function (which corresponds to $k$ in the piecewise polynomial
framework) and a smoothing parameter (the size of the cells in the piecewise
polynomial context and the size of the moving window in the local polynomial
context).  
Of course, the selection of the degree of the polynomial function
is usually easier than the choice of a parameter on a continuous scale
such as $h_n$.  
Nevertheless, our opinion is that kernel estimates are the
most pleasant to use in practice for the following reasons.   
The computation of local and piecewise polynomial estimates 
requires to solve an optimization problem.  
For instance, the computation of piecewise polynomial estimates
is not straightforward, at least for $k>0$.  
When $k=0$, piecewise polynomial estimates reduce to Geffroy's estimate,
whose unsatisfying behavior on finite sample situations 
has been illustrated in Section~\ref{sectionillus}.
At the opposite, kernel and projection estimators enjoy explicit forms and
are thus easily implementable.
Besides, these methods yield smooth estimates whereas piecewise polynomial
estimates are discontinuous whatever the regularity degree of $f$ is. 
Finally, only kernel and projection estimates benefit from
an explicit asymptotic distribution.   This property allows to
build pointwise confident intervals without costly Monte-Carlo
methods.   
In the local polynomial estimates situation (see~\cite{Hall}), 
both the limiting distribution  
and the normalization sequences  
are not explicit making difficult the reduction of the asymptotic bias.

\section{Appendix : proof of lemmas}

\paragraph{\it Proof of Lemma~\ref{lemespe}.}

We give here the complete proof of {\bf (i)} and a sketch of the
proofs of {\bf (ii)} and {\bf (iii)} since the methods in use are
similar.
\begin{description}
\item [(i)] The mathematical expectation can be expanded in three terms:
\begin{eqnarray*}
\E(\xnret-k_n\lnr)=-k_n\lnr e^{-nc\lnr} & + & \int_0^{\mnr} (x-k_n\lnr)F_{n,r}'(x) dx \\
&+& \int_{\mnr}^{\Mnr} (x-k_n\lnr)F_{n,r}(dx).
\end{eqnarray*}
The first term of the sum is asymptotically negligible:
\begin{equation}
\label{an1}
\max_r k_n\lnr e^{-nc\lnr} = \petito{n^{-s}},
\end{equation}
for all $s>0$, when $n\to \infty$.  Using~(\ref{eqfdr}), 
the second term can be rewritten as
\begin{eqnarray*}
\int_0^{\mnr} (x-k_n\lnr)F_{n,r}'(x) dx &=& \frac{nc}{k_n}\int_0^{\mnr} (x-k_n\lnr)
\exp{\left[\frac{nc}{k_n}(x-k_n\lnr)\right]}dx \\
&=& -\frac{k_n}{nc}\int_{\frac{nc}{k_n}(k_n\lnr-\mnr)}^{nc\lnr} u \exp{(-u)} du.
\end{eqnarray*}
Let us note $\phi(u)=(u+1)\exp{(-u)}$ a primitive of $-u\exp{(-u)}$.
We have $\phi(u)=1+\grando{u^2}$ when $u\to 0$ and $\phi(u)=\petito{u^{-s}}$, $\forall s>0$
when $u\to\infty$.
Consequently, remarking that the upper bound goes to infinity, and that the lower
 bound goes to 0 under the assumption $n=\petito{k_n^{1+\alpha}}$ yields
\begin{equation}
\label{an2}
\max_r\Abs{\int_0^{\mnr} (x-k_n\lnr)F_{n,r}'(x) dx +\frac{k_n}{nc}}= \grando{\frac{n}{k_n^{1+2\alpha}}}.
\end{equation}
The third term is bounded above by
$$
\int_{\mnr}^{\Mnr} (x-k_n\lnr)F_{n,r}(dx) \leq (\Mnr-\mnr)
\left[1-\exp{\left(\frac{nc}{k_n}(\mnr-k_n\lnr)\right)}\right]
$$
and thus
\begin{equation}
\label{an3}
\max_r\Abs{\int_{\mnr}^{\Mnr} (x-k_n\lnr)F_{n,r}(dx)}= \grando{\frac{n}{k_n^{1+2\alpha}}}.
\end{equation}
Collecting~(\ref{an1}),~(\ref{an2}) and~(\ref{an3}) proves the result.
\item [(ii)] It is convenient to write the variance as
$$
\Var(\xnret)=\Var(\xnret-k_n\lnr)=\E\left[(\xnret-k_n\lnr)^2\right] - \E^2(\xnret-k_n\lnr).
$$
With a method very similar to the one used to prove {\bf (i)}, we obtain uniformly in $r$,
$$
\E\left[(\xnret-k_n\lnr)^2\right] = \frac{2k_n^2}{n^2c^2} + \grando{\frac{1}{k_n^{2\alpha}}}.
$$
Besides, {\bf (i)} entails
$$
\E^2(\xnret-k_n\lnr)= \frac{k_n^2}{n^2c^2} + \grando{\frac{1}{k_n^{2\alpha}}},
$$
uniformly in $r$, and the conclusion follows.
\item [(iii)] The proof is similar. It requires the calculation of
$\E(\Abs{\xnret-k_n\lnr}^3)$ and the use of {\bf (i)} and {\bf (ii)}.
\end{description}

\paragraph{\it Proof of Lemma~\ref{lemK}.}

Let $\varepsilon>0$ and split~(\ref{eqlemK}) into
\begin{eqnarray*}
\int K(u)K\left(u+\frac{a}{h_n}\right) P_n(du) &=&  
\int_{\Abs{u}>\frac{\Abs{a}}{2 h_n}} K(u)K\left(u+\frac{a}{h_n}\right) P_n(du) \\ &+&
\int_{\Abs{u}\leq\frac{\Abs{a}}{2 h_n}} K(u)K\left(u+\frac{a}{h_n}\right) P_n(du),
\end{eqnarray*}
and consider the two terms separately.
\begin{itemize}
\item The first term is bounded above by
\begin{eqnarray}
\label{eqtmp}
\int_{\Abs{u}>\frac{\Abs{a}}{2 h_n}} K(u)K\left(u+\frac{a}{h_n}\right) P_n(du)
&\leq&  \normsup{K} \int_{\Abs{u}>\frac{\Abs{a}}{2 h_n}}K(u)P_n(du) \\
&=&  \normsup{K} \int_{\Abs{u}>\frac{\Abs{a}}{2 h_n}}\Abs{u}K(u)\frac{1}{\Abs{u}}P_n(du). 
\nonumber
\end{eqnarray}
Since $uK(u)\to 0$ when $\Abs{u}\to\infty$, for $n$ large enough $\Abs{uK(u)}<\varepsilon$ entailing
$$
\int_{\Abs{u}>\frac{\Abs{a}}{2 h_n}} K(u)K\left(u+\frac{a}{h_n}\right) P_n(du) \leq \normsup{K}\varepsilon \int_{\Abs{u}>\frac{\Abs{a}}{2 h_n}} \frac{1}{\Abs{u}}P_n(du) \leq  \frac{2\varepsilon h_n \normsup{K}}{\Abs{a}}.
$$
We have proved that $\forall \varepsilon>0$, for $n$ large enough
$$
\frac{1}{h_n} \int_{\Abs{u}>\frac{\Abs{a}}{2 h_n}} K(u)K\left(u+\frac{a}{h_n}\right) P_n(du) \leq  \frac{2\varepsilon\normsup{K}}{\Abs{a}},
$$
or equivalently,
$$
\int_{\Abs{u}>\frac{\Abs{a}}{2 h_n}} K(u)K\left(u+\frac{a}{h_n}\right) P_n(du) = \petito{h_n}.
$$
\item The second term is bounded above by
\begin{eqnarray*}
\int_{\Abs{u}\leq\frac{\Abs{a}}{2 h_n}} K(u)K\left(u+\frac{a}{h_n}\right) P_n(du)
&\leq&  \normsup{K} \int_{\Abs{u}\leq\frac{\Abs{a}}{2 h_n}} K\left(u+\frac{a}{h_n}\right) P_n(du) \\
&=&\normsup{K} \int_{\Abs{v}>\frac{\Abs{a}}{2 h_n}}K(v)P_n(dv),
\end{eqnarray*}
with $v=u+a/h_n$, and the end of the proof is the same as for~(\ref{eqtmp}).
\end{itemize}

\paragraph{\it Proof of Lemma~\ref{lembiais1}.}

Taking into account that $f$ vanishes out of $[0,1]$, we have
$$
f_n(x)-g_n(x)= \frac{1}{k_n} \somr K_n(x-x_r)f(x_r)-\int_0^1 K_n(x-y)f(y)dy.
$$
Let us define $\phi_{n,x}(y)=K_n(x-y)f(y)$ for $(x,y)\in\R^2$. With this notation,
\begin{equation}
\label{eqinte}
f_n(x)-g_n(x)=\somr \int_{\inr} \left[\phi_{n,x}(x_r)-\phi_{n,x}(y)\right] dy
=\frac{1}{k_n} \somr \int_{-1/2}^{1/2}\left[ \phi_{n,x}(x_r) - \phi_{n,x}\left(x_r+\frac{u}{k_n}\right)\right] du.
\end{equation}
and we have the following expansion:
\begin{eqnarray}
\label{terme1} \phi_{n,x}(x_r) - \phi_{n,x}\left(x_r+\frac{u}{k_n}\right) &=&
K_n\left(x-x_r-\frac{u}{k_n}\right)\left(f(x_r)-f\left(x_r+\frac{u}{k_n}\right)
\right)\\
\label{terme2}
& +&  f(x_r) \left(K_n(x-x_r)-K_n\left(x-x_r-\frac{u}{k_n}\right)\right),
\end{eqnarray}
Now, since $f$ is $\alpha$-Lispchitz, (\ref{terme1}) is uniformly bounded above
by $\normsup{K}L_f/(h_nk_n^\alpha)$.  
The rest of the proof depends on the assumptions made on $K$: 
\begin{description}
\item [(i)] Under \Au, $K$ is $\beta$-Lipschitzian and thus (\ref{terme2})
is uniformly bounded above by 
$\normsup{f} L_K /(h_n^{1+\beta} k_n^\beta)$,
and the conclusion follows.
\item [(ii)] Under \Bu, since $K$ has a compact support, the number
of nonzero terms in (\ref{eqinte}) is $\grando{k_nh_n}$.  
Thus, the contribution of (\ref{terme1}) is $\grando{1/k_n^\alpha}$.  
Two situations have to be considered for the term~(\ref{terme2}).
If $r$ is such that $K$ has only a bounded first derivative  
at $x-x_r$, then (\ref{terme2}) is uniformly bounded above by
$\normsup{f} \normsup{K'} /(h_n^{2} k_n)$. Remarking there are only
a finite number of such terms in (\ref{eqinte}) shows that the
contribution of (\ref{terme2}) is $\grando{1/(h_n^2k_n^2)}$.  
If $r$ is such that $K$ is $C^2$  at $x-x_r$, then 
a second order Taylor expansion yields
$$
K_n(x-x_r)-K_n\left(x-x_r-\frac{u}{k_n}\right)
 = \frac{u}{k_n} K'_n(x-x_r) - \frac{u^2}{2k_n^2} K''_n\left(x-x_r+\theta_n(u)\frac{u}{k_n}\right),
$$
with $\theta_n(u)\in ]0,1[$. Replacing in~(\ref{eqinte}), the first order term 
vanishes, and thus the contribution of (\ref{terme2}) is 
bounded above by $\normsup{f}\normsup{K''_n}h_n/(24k_n^2)$. 
Since $\normsup{K''_n}=\grando{1/h_n^3}$ the result follows.
\end{description}

\paragraph{\it Proof of Lemma~\ref{lembiais2}.}

For any compact subset $C\subset]0,1[$, there exist $0<a<b<1$ such that $C\subset[a,b]$.
Let  $x\in[a,b]$ and consider
\begin{equation}
\label{eqdiff}
g_n(x)-f(x)=  \int_\R K_n(u) (f(x-u)-f(x)) du.
\end{equation}
Let $(\delta_n)$ be a positive sequence tending to 0. Then,~(\ref{eqdiff})
is bounded above by
\begin{eqnarray*}
\Abs{g_n(x)-f(x)}&\leq &\sup_{\Abs{u}\leq \delta_n} \Abs{f(x-u)-f(x)}
+ \int_{\Abs{u}\geq\delta_n} K_n(u) \Abs{f(x-u)-f(x)} du \nonumber\\
& \leq & \sup_{\Abs{u}\leq \delta_n} \Abs{f(x-u)-f(x)} + 2 \normsup{f} \int_{\Abs{u}\geq\delta_n} K_n(u) du.
\end{eqnarray*}
For $n$ large enough, $\delta_n<\min(a,1-b)$ and then $\Abs{u}\leq \delta_n$ entails $(x-u)\in [0,1]$.
Now, since $f$ is $\alpha$-Lipschitzian on $[0,1]$, it yields
\begin{equation}
\label{bochner}
\Abs{g_n(x)-f(x)} \leq L_f \delta_n^\alpha + 2 \normsup{f} \int_{\Abs{u}\geq\delta_n} K_n(u) du.
\end{equation}
Two cases arise:
\begin{description}
\item [(i)] If $u\to u^2 K(u)$ is integrable then
$$
\Abs{g_n(x)-f(x)} \leq L_f \delta_n^\alpha + 2 \normsup{f} \int_{\R} u^2K(u) du \left(\frac{h_n}{\delta_n}\right)^2.
$$
Considering $\delta_n=h_n^{\frac{2}{\alpha+2}}$ in this inequality (which can also be found
page 61 in~\cite{Bosq2} under different hypotheses) gives the result.
\item [(ii)] If $K$ has a compact support, let $A>0$ such that 
supp $(K)\subset [-A,A]$. Then, considering $\delta_n=A h_n$, the
second term in (\ref{bochner}) vanishes and the result is proved.  
\end{description}

\paragraph{\it Proof of Lemma~\ref{lembiais3}.}

Consider $x\in C$. As a consequence of the definitions
\begin{eqnarray*}
\Abs{\E(\fnchapx)-f_n(x)}& \leq & \frac{1}{k_n}\somr K_n(x-x_r) \max_r\Abs{\E(\xnret)-f(x_r)}\\
& \leq & \left(1+\normC{\frac{1}{k_n}\somr K_n(.-x_r)-1} \right) \max_r\Abs{\E(\xnret)-f(x_r)}\\
& \leq & \left(1+\petito{1}\right) \max_r\Abs{\E(\xnret)-f(x_r)},
\end{eqnarray*}
with Corollary~\ref{corosom}. Besides, we have
$$
\Abs{\E(\xnret)-f(x_r)} \leq \frac{k_n}{nc} +\Abs{\E(\xnret)-k_n\lnr + \frac{k_n}{nc}} + \Abs{k_n\lnr-f(x_r)},
$$
and Lemma~\ref{lemespe} yields
$$
\normC{\E(\fnchap)-f_n}=\grando{\frac{k_n}{n}} + \grando{\frac{n}{k_n^{1+2\alpha}}} +\grando{\frac{1}{k_n^\alpha}} = \grando{\frac{k_n}{n}},
$$
under \Ad.  

\paragraph{\it Proof of Lemma~\ref{lemvariance}.}

Let $x\in C$. In view of the independence of the $\xnret$, $r=1,\dots,k_n$,
$$
\Var(\fnchapx)= \frac{1}{k_n^2}\somr K_n^2(x-x_r) \Var(\xnret).
$$
Introducing
$$
\Delta V_n= \frac{n^2}{k_n^2} \max_r \Abs{\Var(\xnret)- \frac{k_n^2}{n^2 c^2}}
\mbox{ and }
\Delta K_n= \normC{\frac{1}{h_nk_n}\somr K^2\left(\frac{.-x_r}{h_n}\right)-\normdeux{K}^2},
$$
we have
$$
\normC{\frac{\Var(\fnchap)}{\sigma_n^2} - \sigma^2} \leq \Delta K_n(\Delta V_n + 1/c^2)+ \Delta V_n \normdeux{K}^2.
$$
Lemma~\ref{lemespe} shows that $\Delta V_n\to 0$,  Corollary~\ref{corosom} 
applied to the kernel $K^2/\normdeux{K}^2$ shows that $\Delta K_n\to 0$
as $n\to\infty$, and the conclusion follows.

\paragraph{\it Proof of Lemma~\ref{newbias}.}

The bias expands as
$$
\normC{\E(\fntild)-f}\leq  \normC{f_n-f} + \normC{\E(\fntild)-f_n},
$$
which first term is controlled by Proposition~\ref{propbiais}. Consider the second term:
\begin{eqnarray}
\Abs{\E(\fntildx)-f_n(x)}&\leq&\left(\frac{1}{k_n}\somr K_n(x-x_r)\right)\max_r\Abs{\E(\xnret)+\E(Z_n)-f(x_r)}\\
& \leq & \left(1+\normC{\frac{1}{k_n}\somr K_n(.-x_r)-1}\right)\max_r\Abs{\E(\xnret)+\E(Z_n)-f(x_r)}. \nonumber
\end{eqnarray}
Corollary~\ref{corosom} shows that it is sufficient to consider 
$$
\Abs{\E(\xnret)+\E(Z_n)-f(x_r)} \leq \Abs{\E(\xnret)-k_n\lnr+\frac{k_n}{nc}} + \Abs{\E(Z_n)-\frac{k_n}{nc}}
+ \Abs{k_n\lnr-f(x_r)}.
$$
Lemma~\ref{lemespe} and~(\ref{espezn}) yield
$$
\normC{\E(\fntild)-f_n} = \grando{\frac{n}{k_n^{1+2\alpha}}} + 
\grando{\frac{1}{k_n^{2\alpha}}} + \grando{\frac{1}{k_n^{\alpha}}} 
= \grando{\frac{1}{k_n^{\alpha}}} 
$$
under \Bd, and the conclusion follows.    


\end{document}